\documentclass[11pt]{article}

\usepackage{natbib}
 \bibpunct[, ]{(}{)}{,}{a}{}{,}%
 \def\bibfont{\small}%
 \def\bibsep{\smallskipamount}%
 \def\bibhang{24pt}%
\usepackage{epsfig,amsfonts,latexsym,amsmath,tikz,fullpage}

\newtheorem{theorem}{Theorem}
\newtheorem{corollary}{Corollary}
\newtheorem{lemma}{Lemma}

\newtheorem{conjecture}{Conjecture}

\newenvironment{myproof}{\noindent\textbf{Proof\ }}{\hspace*{\fill}$\Box$\medskip}

\begin{document}

\title{The local-global conjecture for scheduling with non-linear cost}

\author{
Nikhil Bansal%
\thanks{Department of Mathematics and Computer Science, Eindhoven University of Technology, 5600 MB Eindhoven, The Netherlands \textsl{n.bansal@tue.nl}}
\and
Christoph D\"urr%
\thanks{Sorbonne Universit\'es, UPMC Univ Paris 06, CNRS, LIP6, 75005 Paris, France \textsl{christoph.durr@lip6.fr}}
\and
Nguyen Kim Thang%
\thanks{IBISC, University Evry Val d'Essonne, 23 blv. de France, 91034 Evry, France \textsl{thang@ibisc.fr}}
\and
\'Oscar C. V\'asquez%
\thanks{Universidad de Santiago, Departamento de Ingenier\'ia Industrial, 3769 Av. Ecuador, Santiago, Chile \textsl{oscar.vasquez@usach.cl}}
} 

\maketitle

\begin{abstract}
  We consider the classical scheduling problem on a single machine, on which we need to
  schedule sequentially $n$ given jobs.  Every job $j$ has a processing time $p_j$ and a priority weight $w_j$, and for a given schedule a completion time $C_j$.  In this paper we consider the problem of minimizing the objective value $\sum_j w_j C_j^\beta$ for some fixed constant $\beta>0$.
  This non-linearity is motivated for example by the learning effect of a machine improving its efficiency over time, or by the speed scaling model.
	For $\beta=1$, the well-known Smith's rule that orders job in the non-increasing  order  of $w_j/p_j$ give the optimum schedule. However, for $\beta \neq 1$,  the complexity status
of  this problem  is  open. Among other things, a key issue here is that the ordering between a pair of jobs is not well-defined, and might depend on where the jobs lie in the schedule and also on the jobs between them.
We investigate this question systematically and substantially generalize the previously known results in this direction. These results lead to interesting new dominance properties among schedules which lead to huge speed up in exact algorithms for the problem. An experimental study evaluates the  impact of these properties on
the exact algorithm A*.


\end{abstract}

\maketitle

\section{Introduction}

In a typical scheduling problem we have to order $n$ given jobs, each with a different processing time, so to minimize some problem specific cost function.  Every job $j$ has a positive processing time $p_j$ and a priority weight $w_j$.  A schedule is defined by a ranking $\sigma$, and the completion time of job $j$ is defined as $C_j := \sum_i p_i$, where the sum ranges over all jobs $i$ such that $\sigma_i\leq \sigma_j$. The goal is to produce a schedule that minimizes some cost function involving the jobs' weights and the completion times.

A popular objective function is the weighted average completion time $\sum w_jC_j$ (omitting the normalization factor $1/n$).  It has been known since the 1950's that optimal schedules are precisely the orders following an decreasing \emph{Smith-ratio} $w_j/p_j$, as has been shown by a simple exchange argument \citep{Smith:56:Linear-penalty-function}.

In this paper we consider the more general objective function $\sum w_j C_j^\beta$ for some constant $\beta>0$, and denote the problem by $1||\sum w_j C_j^\beta$.  Several motivations have been given in the literature for this objective.  For example it can model the standard objective $\sum w_jC_j$, but on a machine changing its execution speed continuously.  This could result from a learning effect, or the continuous upgrade of its resources, or from a wear and tear effect, resulting in a machine which works less effective over time.  A recent motivation comes from the speed scaling scheduling model.
In
\citet{DurrJezVasquez:14:Optimization-Mechanism},
and \citet{MegowVerschae:13:PTAS-piecewise-penalty-function} the problem of minimizing total weighted completion time plus total energy consumption was studied, and both papers reduced this problem to the problem considered in this paper for a constant $1/2\leq\beta\leq 2/3$.  However as we mention later in the paper, most previous research focused on the $\beta=2$ case, as the objective function then represents a trade off between maximum and average weighted completion time.

\section{Dominance properties}

The complexity status of the problem $1||\sum w_j C_j^\beta$ is open for $\beta\neq1$ in the sense that neither  polynomial time algorithms nor NP-hardness proofs are known.  For $\beta=1$ the problem is polynomial, as has been shown by a simple exchange argument. When $i,j$ are adjacent jobs in a schedule, then the order $ij$ is preferred over $ji$ whenever $w_i/p_i > w_j/p_j$
and that is independent from all other jobs.  In this case we denote this property by $i\prec_\ell j$.  Assume for simplicity that all jobs $k$ have a distinct ratio $w_k/p_k$, which is called the \emph{Smith-ratio}.  Under this condition $\prec_\ell$ defines a total order on the jobs, that leads to the unique optimal schedule.

For general $\beta$ values, the situation is not so simple, as in term of objective cost the effect of exchanging two adjacent jobs depends on their position in the schedule.  So for two jobs $i,j$ none of $i\prec_\ell j, j\prec_\ell i$ might hold, which is precisely the difficulty of this scheduling problem.

However it would be much more useful if for some jobs $i,j$ we know that an optimal schedule always schedules $i$ before $j$, no matter if they are adjacent or not.  This property is denoted by $i\prec_g j$. Having many pairs of jobs with such a property could dramatically help  in improving exhaustive search procedures to find an exact schedule.  Section~\ref{sec:experimental} contains an experimental study on the impact of this information on the performance of some search procedure.

Therefore several attempts have been proposed to characterize the property $i\prec_g j$ as function of the job parameters $p_i,w_i,p_j,w_j$ and of $\beta$.  Several sufficient conditions have been proposed, however they are either far from what is necessary, or are tight only in some very restricted cases~:

\begin{description}
    \item[\textbf{Sen-Dileepan-Ruparel \citep{SenDileepan:90:Order-constraint-generalized-quadratic-penalty-function}}] for any $\beta>0$, if $w_i>w_j$ and $p_i\leq p_j$, then $i\prec_g j$.
    \item[\textbf{Mondal-Sen-Höhn-Jacobs-1 \citep{HohnJacobs:12:Experimental-analytical-quadratic-penalty-function}}] for $\beta=2$, if $w_i/p_i > \beta w_j/p_j$,  then $i\prec_g j$.
    \item[\textbf{Mondal-Sen-Höhn-Jacobs-2 \citep{HohnJacobs:12:Experimental-analytical-quadratic-penalty-function}}] for $\beta=2$, if $w_i\geq w_j$ and $w_i/p_i > w_j/p_j$,  then $i\prec_g j$.
\end{description}




\subsection*{Related work}
\label{sec:related-work}

Embarrassingly, very little is known about the computational complexity
of  this  problem,  except  for  the  special  case  $\beta=1$ which was solved in the 1950's
\citep{Smith:56:Linear-penalty-function}.  In that case
scheduling jobs in order of decreasing Smith ratio $w_j/p_j$ leads to
the  optimal   schedule.

Two research  directions were  applied to this  problem, approximation
algorithms and branch  and bound algorithms.  The former
have been  proposed for the  even more general problem  $1||\sum f_j(C_j)$,
where every job $j$ is  given an increasing penalty function $f_j(t)$,
that does not need to be  of the form $w_j t^\beta$.
A constant factor approximation algorithm has been proposed by
  \citet{BansalPruhs:10:Approach-geometry-scheduling-penalty-function-of-jobs}  based on  a  geometric interpretation of the problem.
The approximation factor  has been  improved  from $16$  to  $2  +\epsilon$ via  a
primal-dual       approach       by       \citet{CheungShmoys:11:Approach-primal-dual-scheduling-penaty-function-of-jobs}.   The  simpler problem  $1||\sum   w_j  f(C_j)$ was considered in  \citet{EpsteinLevin:10:Approach-adversary-scheduling-weighted-penalty-function-of-jobsUniversal-sequencing}, who   provided  a
$4+\epsilon$ approximation  algorithm for the setting where  $f$ is an
arbitrary  increasing differentiable  penalty function  chosen  by the
adversary \emph{after}  the schedule  has been produced.  A polynomial
time   approximation   scheme  has   been   provided  by \citet{MegowVerschae:13:PTAS-piecewise-penalty-function}  for the problem $1||\sum   w_j  f(C_j)$ , where $f$ is an arbitrary monotone penalty function.

Finally, \citet{HohnJacobs:12:Performance-Smith-rule} derived  a method to
compute the tight approximation factor of the Smith-ratio-schedule for
any particular  monotone increasing  convex or concave  cost function.
In particular  for $f(t)=t^\beta$ they obtained for  example the ratio
$1.07$ when $\beta=0.5$ and the ratio $1.31$ when $\beta=2$.

Concerning branch-and-bound algorithms, several papers give sufficient conditions for the global order property, and analyze experimentally the impact on branch and bound algorithms of their contributions.
Previous research  focused mainly  on the quadratic  case $\beta=2$, see
\citet{Townsend:78:Order-constraint-quadratic-penalty-function,BaggaKarlra:80:Order-constraint-quadratic-penalty-function,SenDileepan:90:Order-constraint-generalized-quadratic-penalty-function,Alidaee:93:Order-constraint-for-non-linear-penalty-function,CroceTadeiBarracoDiTullio:93:Order-constraint-generalized-quadratic-penalty-function,Szwarc:98:Decomposition-for-non-linear-penalty-function}.
\citet{MondalSen:00:Conjecture-order-constraint-quadratic-penalty-function}
conjectured that $\beta=2 \wedge (w_i\geq w_j) \wedge (w_i/p_i > w_j/p_j)$
implies  the   global  order  property  $i\prec_g j$,  and  provided
experimental evidence that this property would significantly improve
the runtime of a branch-and-prune search.  Recently, \citet{HohnJacobs:12:Experimental-analytical-quadratic-penalty-function} succeeded to prove this  conjecture.  In  addition they  provided a weaker sufficient condition for $i\prec_g j$ which holds for any integer $\beta\geq 2$.
An extensive experimental study analyzed the effect of these results on the performance of the branch-and-prune search.


\section{Our contribution}

All previously proposed sufficient conditions for ensuring that $i\prec_g j$ were rather ad-hoc, and are much stronger than what seems to be necessary.
So this motivated our main goal of obtaining a precise characterization of $i\prec_g j$, for each value $\beta>0$.

In contrast the condition $i\prec_\ell j$ is fairly easy to characterize, using simple algebra, as has been described in the past by \citet{HohnJacobs:12:Experimental-analytical-quadratic-penalty-function} for $\beta=2$.  This characterization holds in fact for any value of $\beta$ and for completeness we describe it in Section~\ref{sec:local}.

As $i\prec_g j$ trivially implies $i\prec_\ell j$, the strongest (best) possible result one could hope for is that  $i \prec_g j$ occurs precisely when $i\prec_\ell j$.
If true, this would give to a local exchange property a broader impact on the structure of optimal schedules, and have a strong implication on the effect of non-local exchanges.

Having observed the optimal solutions of a large set of instances, this property seems to be the right candidate for a characterization. Moreover, this was also suggested by previous results for particular cases.  For example  \citet{HohnJacobs:12:Experimental-analytical-quadratic-penalty-function}  showed that if $\beta=2$ and $p_i \leq p_j$ then $i\prec_g j$ if and only if $i\prec_\ell j$.
The same characterization has been shown for a related objective function, where one wants to maximize $\sum w_j C_j^{-1}$ \citep{Vasquez:14:For-the-airplane-refueling}.

This situation motivates us to state the following conjecture.
\begin{conjecture}\label{conjecture}[ {\bf Local-Global Conjecture}]
  For any $\beta>0$ and all jobs $i,j$, $i\prec_g j$ if and only if $i\prec_\ell j$.
\end{conjecture}

We succeed to show this claim in the case $\beta\geq 1$.
Somewhat surprisingly, the proof turns out to be extremely subtle and involved. In particular, it requires the use of several non-trivial properties of polynomials and carefully chosen inequalities among them, and then finally combining them using a carefully chosen weighted combination.
Our proof distinguishes the cases $p_j< p_i$ and $p_j \geq p_i$. The first case is substantially easier than the second one. In fact, in the first case we can show that local-global conjecture for every $\beta >0$. However, for the second case ($p_j \geq p_i$) when $0<\beta<1$ we only give a necessary condition for $i\prec_g j$.

While these results do not tackle the problem of the computational complexity of the problem, they nevertheless provide a deeper insight in its structure, and in addition speed up exhaustive search techniques in practice.
This is due to the fact that with the conditions for $i\prec_g j$ provided in this paper it is now possible to conclude $i\prec_g j$ for a significant portion of job pairs, for which previously known conditions failed.
In the final Section \ref{sec:experimental} of this paper, we study experimentally the impact of our contributions on the procedure A* for this problem. Improvements in the running time by a factor of 1000 or more  have been observed for some random instances (see  Section~\ref{subsec:gener}).

\begin{figure*}[htb]
\begin{tabular}{ccc}
\huge $0<\beta<1$
&
\huge $\beta>1$
&
\huge $\beta=2$
\\
\raisebox{-.5\height}{\includegraphics[width=5cm]{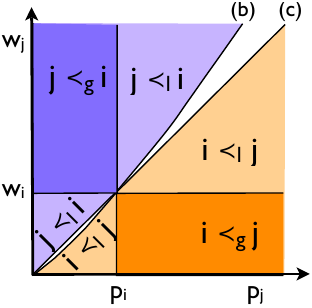}}
&
\raisebox{-.5\height}{\includegraphics[width=5cm]{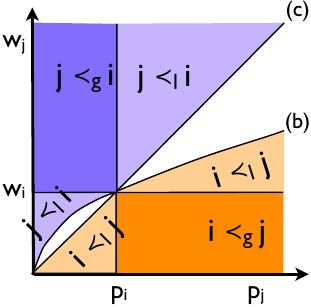}}
&
\raisebox{-.5\height}{\includegraphics[width=5cm]{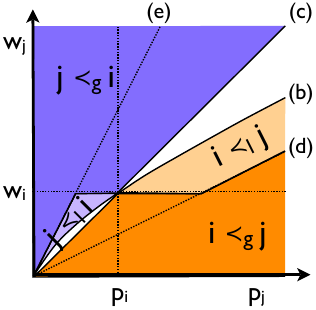}}
\\
\raisebox{-.5\height}{\includegraphics[width=5cm]{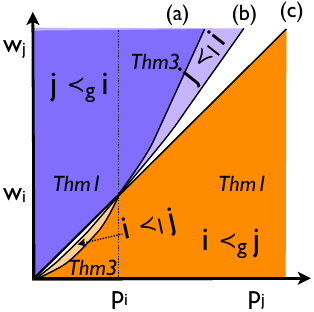}}
&
\raisebox{-.5\height}{\includegraphics[width=5cm]{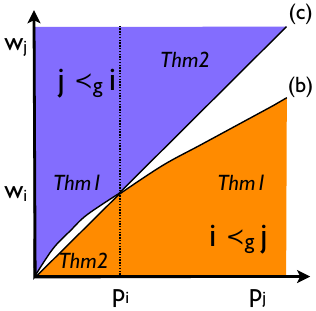}}
&
\begin{minipage}{4cm}
(a) $w_j= w_i (p_j/p_i)^{2-\beta}$
\\
(b) $w_j= w_{i} \frac{(p_i+p_j)^\beta-p_i^\beta}{(p_i+p_j)^\beta-p_j^\beta}$ 
\\
(c) $w_j=w_i p_j /p_i$
\\
(d) $w_j=w_i p_j  /2p_i$
\\
(e) $w_j=2 w_i p_j  /p_i$
\end{minipage}
\end{tabular}
\caption{Illustration of our contribution (bottom row) compared to previous results (top row), using a similar representation as in \citep{HohnJacobs:12:Experimental-analytical-quadratic-penalty-function}.  Every point in the diagram represents a job $j$ with respect to some fixed job $i$. The space is divided into regions where $i\prec_\ell j$ holds or $j\prec_\ell i$ or none. These regions contain subregions where we know that the stronger condition $\prec_g$ holds. The boundaries are defined by functions which are named from (a) to (e). The last 2 diagrams also indicate the related theorems.}
\label{fig:contribution}
\end{figure*}

\section{Technical lemmas}
\label{sec:techlem}

This section contains several technical lemmas used in the proof of our main theorems.

\begin{lemma}
 \label{lem:wf}
For $0 <\beta< 1, a < b$ and $p_i>p_j$,
 \[\left(\frac{p_i}{p_j}\right)^{1-\beta} \cdot \frac{f (b + p_i) - f (a + p_i)}{f (b + p_j) - f       (a + p_j)} \geq 1. \]
\end{lemma}
\begin{myproof}
  For this purpose we define the function
  \[ g (x) := x^{1-\beta} (f (b + x) - f (a + x)) \]
  and show that $g$ is increasing, which implies $g (p_i) / g (p_j) \geq 1$
    as required. So we have to show $ g' (x) >0$ in other words
    \begin{align*}
       \phantom{+}(1-\beta) x^{-\beta} \big(f (b + x) - f (a + x)\big) \\
       + x^{1-\beta} \big(f' (b + x) - f' (a + x)\big)
      &\geq 0
      \\
      \Leftrightarrow
      \phantom{+}(1-\beta) x^{-\beta} \big((b + x)^\beta - (a + x)^\beta\big)
      \\+ x^{1-\beta} \beta \big((b + x)^{\beta-1} - (a + x)^{\beta-1}\big)
     & \geq 0
      \\
      \Leftrightarrow
       \phantom{-}(b + x)^{\beta-1}((1-\beta)(b + x)+\beta x)
       \\-(a + x)^{\beta-1}((1-\beta)(a + x)+\beta x)
       &\geq 0 .
      \end{align*}
 To establish the last inequality, we introduce another function
    \[ r (z):= (z + x)^{\beta-1}((1-\beta)(z + x)+\beta x) \]
    and show that $r(z)$ is increasing, implying $r (b) \geq r (a)$. By
    analyzing its derivative we obtain
    \begin{align*}
      r' (z) =&  \phantom{+} (\beta-1) (z + x)^{\beta-2}((1-\beta)(z + x)+\beta x)
      \\&+
      (1-\beta)(z + x)^{\beta-1}\\
      = & (z + x)^{\beta - 2}(1-\beta)((\beta-1)(z + x) - x \beta + z+x )\\
       = & (z + x)^{\beta - 2}(1-\beta) (\beta z),
    \end{align*}
    which is positive as required. This concludes the proof. 
\end{myproof}

Some of our proofs are based on particular properties which are enumerated in the following lemma.
\begin{lemma}  \label{lem:log-concave}
 The function $f(t)=t^\beta$ defined for $\beta \geq 1$ satisfies the following properties.
\begin{enumerate}
\item $f(x) \geq 0$ for  $x\geq 0$.
\item $f$ is convex and non-decreasing, i.e. $f',f'' \geq 0$.
\item $f'$ is log-concave (i.e. $\log(f')$  is concave), which implies that $f''/f'$ is non-increasing.
Intuitively this means that $f$ does not increase much faster than  $e^x$.
\item For every $b>0$, the function $g_b(x)= f(b+e^x)-f(b)$ is log-convex in $x$.
Intuitively this means that $f(b+e^x)-f(b)$ increases faster than $e^{cx}$ for some $c>0$. Formally this means
\begin{equation}
\label{complex:cond}
y f'(b+y)/\big(f(b+y)-f(b)\big)
\end{equation}
 is increasing in $y$.
\end{enumerate}
\end{lemma}
The proof is based on standard functional analysis and is omitted.

\begin{lemma}
\label{lem:logcon}
For $a< b$, the fraction
$$ \frac{f'(b) - f'(a) }{f(b)-f(a)} $$
\begin{itemize}
\item is decreasing in $a$ and decreasing in $b$ for any $\beta>1$
\item and is increasing in $a$ and increasing in $b$ for any $0<\beta<1$.
\end{itemize}
\end{lemma}
\begin{myproof}
First we consider the case $\beta>1$. We can write $f'(b) - f'(a) = \int_{a}^b f''(x)dx$ and $f(b)-f(a) = \int_a^b f'(x)dx$. Note that $f''$ and $f'$ are non-negative.
By $\beta>1$ and Lemma~\ref{lem:log-concave} $f'$ is log-concave, which means that $f''(x)/f'(x)$  is non-increasing in $x$. This implies
\begin{align*}
  \int_a^b \frac{f''(b)}{f'(b)} f'(x) dx
  &\leq
  \int_a^b \frac{f''(x)}{f'(x)} f'(x) dx
  \\
 &\leq
  \int_a^b \frac{f''(a)}{f'(a)} f'(x) dx
\end{align*}
Hence
\[
\frac{ \int_a^b f''(b) dx }{ \int_a^b f'(b) dx }
\leq
\frac{ \int_a^b f''(x) dx }{ \int_a^b f'(x) dx }
\leq
\frac{ \int_a^b f''(a) dx }{ \int_a^b f'(a) dx }.
\]
For positive values $u_1,u_2,u_3,v_1,v_2,v_3$ with $u_1/v_1 \leq u_2/v_2 \leq u_3/v_3$ we have
\[
\frac{ u_1+u_2 }{ v_1+v_2 }
\leq
\frac{ u_2 }{ v_2 }
\leq
\frac{ u_2+u_3 }{ v_2+v_3 }.
\]
We use this property for $a'<a<b<b'$
\begin{align*}
\frac{u_1}{v_1} &= \frac{ \int_{b}^{b'} f''(x) dx }{ \int_{b}^{b'} f'(x) dx },
\\
\frac{u_2}{v_2} &= \frac{ \int_{a}^{b} f''(x) dx }{ \int_{a}^{b} f'(x) dx },
\\
\frac{u_3}{v_3} &= \frac{ \int_{a'}^{a} f''(x) dx }{ \int_{a'}^{a} f'(x) dx }
\end{align*}
and obtain
\[
\frac{ \int_{a'}^{b} f''(x) dx }{ \int_{a'}^{b} f'(x) dx }
\leq
\frac{ \int_{a}^{b} f''(x) dx }{ \int_{a}^{b} f'(x) dx }
\leq
\frac{ \int_{a}^{b'} f''(x) dx }{ \int_{a}^{b'} f'(x) dx }.
\]

For the case $0<\beta<1$ the argument is the same using the fact that $f''(x)/f'(x)$  is non-decreasing in $x$.
\end{myproof}

The previous lemma permits to show the following corollary.
\begin{corollary}
    \label{lem:q}For $t \geq 0$ let the function $q$ be defined as
    \[ q (t):= \frac{f (t + p_j) - f (t)}{f (t + p_i) - f (t)}.
    \]
    For $p_i>p_j$, if $\beta > 1$ then $q$ is increasing and if $0 < \beta < 1$ then $q$ is
    decreasing.
  \end{corollary}
\begin{myproof}
We only prove the case $\beta>1$, the other case is similar.
Showing that $q(t)$ is increasing, it suffices to show that
\[
\ln q(t) = \ln(f (t + p_j) - f (t)) - \ln(f (t + p_i) - f (t))
\] is increasing.
To this purpose we notice that the derivative
\[
\frac{f'(t + p_j) - f'(t))}{f (t + p_j) - f (t))}
-
\frac{f'(t + p_i) - f'(t))}{f (t + p_j) - f (t))}
\]
is positive by Lemma~\ref{lem:q} and $p_i>p_j$.
\end{myproof}

\begin{lemma}
\label{inc}
If  $\beta>1, a < b$ and  $$g(x) = x \frac{f(b+x) -f(x+a)}{f(b+x) - f(b)},$$ then $g(x)$ is a non-decreasing function of $x$.
\end{lemma}
\begin{myproof}
Equivalently, we show that $\ln g(x) = \ln(x) + \ln \big(f(b+x) - f(a+x) \big)  - \ln \big(f(b+x)-f(b)\big) $ is non-decreasing.
Taking derivative of the right hand side, we show
$$ \frac{1}{x} +  \frac{f'(b+x) - f'(a+x)}{f(b+x) - f(a+x)}   - \frac{f'(b+x)}{f(b+x)-f(b)} \geq 0.$$
By log-concavity of $f'$ and Lemma \ref{lem:logcon}, the second term is minimized when $a$ approaches $b$, and hence is at least $f''(b+x)/f'(b+x)$.
Therefore it is enough to show that
$$ \frac{1}{x} +  \frac{f''(b+x)}{f'(b+x)}   - \frac{f'(b+x)}{f(b+x)-f(b)} \geq 0,$$
which is equivalent in showing that
$$ \ln (x) + \ln (f'(b+x)) - \ln \big(f(b+x)-f(b)\big)$$ is non-decreasing $x$.
The later derives from the fact that
$ x f'(b+x)/\big(f(b+x)-f(b)\big)$ is non-decreasing, which follows from assumption in \eqref{complex:cond}.
\end{myproof}

\section{Characterization of the local order property}
\label{sec:local}

To simplify notation, throughout the  paper we assume that no two jobs
have  the same  processing time,  weight or  Smith-ratio  (weight over
processing  time).  The proofs extend to the general case by considering an additional
tie-breaking rule between jobs with identical parameters.
For convenience  we  extend  the  notation of  the
penalty function $f$ to the makespan of schedule $S$ as $f(S):=f(\sum_{i\in S}
p_i)$. Also we denote by $F(S)$ the cost of schedule $S$.

In order to analyze the effect of exchanging adjacent jobs,
we define the following function on $t\geq 0$
\[\phi_{ij}(t):=\frac{f(t+p_i+p_j)-f(t+p_j)}{f(t+p_i+p_j)-f(t+p_i)}.\]
Note  that  $\phi_{ij}(t)$  is  well  defined since  $f$  is  strictly
increasing by assumption and the durations $p_i,p_j$ are non-zero.  This function $\phi_{ij}$   permits us to analyze algebraically
the local order property, since
\begin{align} \label{eq:localphi}
i \prec_{\ell} j\:\Leftrightarrow\: \forall t\geq0: \phi_{ij}(t) < \frac{w_i}{w_j}.
\end{align}

The  following technical  lemmas  show properties  of $\phi_{ij}$  and
relate them to properties of $f$.

\begin{lemma}\label{lem:monotone}
If $p_i\neq p_j$ then $\phi_{ij}$ is strictly monotone, in particular:
\begin{itemize}
\item If $p_i>p_j$ and $\beta>1$, then $\phi_{ij}$ is strictly increasing.
\item If $p_i<p_j$ and $\beta>1$, then $\phi_{ij}$ is strictly decreasing.
\item If $p_i>p_j$ and $\beta<1$, then $\phi_{ij}$ is strictly decreasing.
\item If $p_i<p_j$ and $\beta<1$, then $\phi_{ij}$ is strictly  increasing.
\end{itemize}
\end{lemma}
\begin{myproof}
We only show this statement for the first case $p_i>p_j$ and $\beta>1$, and the other cases are similar.
In order to show that $\phi_{ij}$ is strictly increasing we prove that $\ln \phi_{ij}$ is increasing. For this we analyze its derivative which is
\[
\frac{ f'(t+p_i+p_j)-f'(t+p_j)}{f(t+p_i+p_j)-f(t+p_j)} - \frac{ f'(t+p_i+p_j)-f'(t+p_i)}{f(t+p_i+p_j)-f(t+p_i)}.
\]
The derivative is positive by Lemma~\ref{lem:logcon}.
\end{myproof}

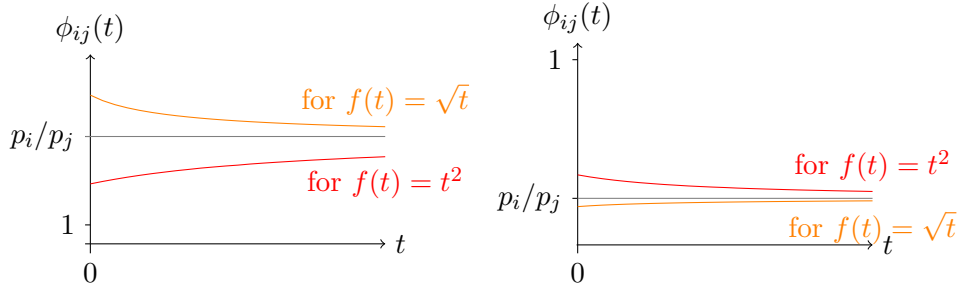
\begin{figure}[ht]
\begin{tikzpicture}[scale=0.14,y=1.8cm]
  \draw[->] (19.5,0) -- (48,0) node[right] {$t$};
  \draw[->] (20,0) -- (20,10) node[above] {$\phi_{ij}(t)$};
  \draw[-](20,0)--(20,-0.5) node[below] {$0$};
  \draw[-] (20,17./3)--(19.5,17./3) node[left] {$p_i/p_j$};
  \draw[-] (20,1)--(19.5,1) node[left] {1};
  \draw[color=gray,domain=19.5:48]  plot (\x, 17.0/3) ;
  \draw[color=orange,domain=20:48]    plot   (\x,{(sqrt(\x)-sqrt(\x-17))/(sqrt(\x)-sqrt(\x-3))})
  node[above] {for $f(t) = \sqrt t$};
  \draw[color=red,domain=20:48] plot (\x,{ (pow(\x,2.3)-pow(\x-17,2.3)) /
    (pow(\x,2.3)-pow(\x-3,2.3)) })
  node[below] {for $f(t) = t^2$};
\end{tikzpicture}
\begin{tikzpicture}[scale=0.14,y=16cm]
  \draw[->] (19.9,-0.1) -- (48,-0.1) node[right] {$t$};
  \draw[->] (20,-0.1) -- (20,1.1) node[above] {$\phi_{ij}(t)$};
  \draw[-](20,-0.1)--(20,-0.15) node[below] {$0$};
  \draw[-] (20,3/17.)--(19.5,3/17.) node[left] {$p_i/p_j$};
  \draw[-] (20,1)--(19.5,1) node[left] {1};
  \draw[color=gray,domain=20:48]  plot (\x, 3/17.) ;
  \draw[color=orange,domain=20:48]    plot   (\x,{(sqrt(\x)-sqrt(\x-3))/(sqrt(\x)-sqrt(\x-17))})
  node[below] {for $f(t) = \sqrt t$};
  \draw[color=red,domain=20:48] plot (\x,{ (pow(\x,2.3)-pow(\x-3,2.3)) /
    (pow(\x,2.3)-pow(\x-17,2.3)) })
  node[above] {for $f(t) = t^2$};
\end{tikzpicture}
\caption{Examples of  the function $\phi_{ij}(t)$  for $\beta=0.5$ and
  $\beta=2$, as well as for the cases $p_i>p_j$ and $p_i<p_j$.}
\label{fig:sqrt}
\end{figure}

\begin{lemma}\label{lem:limit}
For any jobs $i,j$, we have
$         \lim_{t\rightarrow \infty} \phi_{ij}(t) = p_i/p_j.$
\end{lemma}
\begin{myproof}
By the mean value Theorem, for any differentiable function $f$ and $y,x$ it holds that $f(y) = f(x) + (y-x) f'(\eta)$ for some $\eta \in [x,y]$.
Thus $(t+p_i+p_j)^\beta - (t+p_i)^\beta = \beta p_j (t+p_i+\eta)^{\beta-1}$ for some $\eta \in [0,p_j]$ and
$(t+p_i+p_j)^\beta-(t+p_j)^\beta = \beta p_i(t+p_j+\gamma)^{\beta-1}$ for some $\gamma \in [0,p_i]$.
Moreover, for any $\beta >0$ and $a>0$, $\lim_{t \rightarrow \infty} (t+a)^{\beta-1}/t^{\beta-1} =  1$.
Therefore,
$$\lim_{t \rightarrow \infty} \phi_{ij}(t)=\lim_{t \rightarrow \infty} \frac{(t+p_i+p_j)^\beta-(t+p_j)^\beta}{(t+p_i+p_j)^\beta-(t+p_i)^\beta} = \frac{p_i}{p_j}.$$


\end{myproof}

These two lemmas permit to characterize the local order property, see Figure~\ref{fig:contribution}.
\begin{lemma} \label{lem:local}
  For any two jobs $i,j$ we have $i\prec_\ell j$ if and only if
  \begin{itemize}
    \item $\beta> 1$ and $p_i\leq p_j$ and $w_j/w_i\leq \phi_{ji}(0)$ or
    \item $\beta> 1$ and $p_i\geq p_j$ and $w_j/w_i\leq p_j/p_i$ or
    \item $0<\beta < 1$ and $p_i\leq p_j$ and $w_j/w_i\leq p_j/p_i$ or
    \item $0<\beta < 1$ and $p_i\geq p_j$ and $w_j/w_i\leq \phi_{ji}(0)$.
  \end{itemize}
\end{lemma}

\section{The global order property}

In this section we characterize the global order property of two jobs $i,j$ in the convex case $\beta>1$, and provide sufficient conditions on the concave case $0\leq\beta<1$.  Our contributions are summarized graphically in Figure~\ref{fig:contribution}.

\subsection{Global order property for $p_i\leq p_j$}
\label{sec:case1}

In this section we give the proof of the conjecture in case $i$ has processing time not larger than $j$. Intuitively this seems the easier case, as exchanging $i$ with $j$ in the schedule $AjBi$ makes jobs from $B$ complete earlier.  However the benefit of the exchange on these jobs cannot simply be ignored in the proof.  A simple example with $\beta = 2$ shows why this is so.  Let $i,j,k$ be 3 jobs with $p_i=4,w_i=1,p_j=8,w_j=1.5,p_k=1,w_k=0$.
Then $i\prec_\ell j$,
but exchanging $i,j$ in the schedule $jki$ increases the objective value, as
$F(ikj)=4^2 + 1.5 \cdot 13^2 = 269.5$ while $F(jki)=1.5\cdot 8^2 + 13^2=265$.
Now if we raise $w_k$ to $0.3$, then we obtain an interesting instance. It satisfies $F(jki)<F(jik)$ and $jki$ is the optimal schedule, but it cannot be shown with an exchange argument from $ikj$ without taking into account the gain on job $k$ during the exchange.

\begin{theorem}~\label{thm:lg-big}
The implication $i\prec_\ell j \Rightarrow i\prec_g j$ holds when $p_i\leq p_j$.
\end{theorem}
\begin{myproof}
The proof holds in fact for any increasing penalty function $f$.
Let $A,B$ be two arbitrary job sequences. We will show that the schedule $AjBi$ has strictly higher cost than one of the schedules $AijB, AiBj$.

First if $F(AjBi) \geq F(AjiB)$, then by $i\prec_\ell j$ we have even $F(AjBi)>F(AijB)$.
So for the remaining case we assume $F(AjBi)<F(AjiB)$ and will show $F(AjBi) > F(AiBj)$. By $i\prec_\ell j$ it would be enough to show the even stronger inequality
\begin{align*}
  F (A j i B) - F (A i j B)
&<
      F (A j B i) - F (A i B j),
\end{align*}
or equivalently
\begin{align*}
  F (A j i B) - F (A j B i)
  &<
      F (A i j B) - F (A i B j).
\end{align*}
The left hand side is positive by assumption, so it would be enough to show
\begin{align}
 \min_t \phi_{ji}(t) & \left( F (A j i B) - F (A j B i) \right) \notag
  \\
  <&
      ~F (A i j B) - F (A i B j),
  \label{eq:thm1:goal}
\end{align}
since $\phi_{ji}(t)>1$ by $p_i<p_j$.

We introduce the following notation.  Denote the jobs in $B$ by $1,\ldots,k$ and for every job $1\leq h\leq k$ denote by $l_h$ the total processing time of all jobs from $1$ to $h$.
We show the inequality, by analyzing separately the contribution of jobs $h\in B$, and of the jobs $i,j$.
By definition of $\phi_{ji}$ we have
\begin{align*}
&\phi_{ji}(l_h)\left( f(p_i+p_j+l_h) - f(p_j+l_h) \right)
\\=&
f(p_i+p_j+l_h)-f(p_i+l_h) ,
\end{align*}
which implies
\begin{align}
&\min_t \phi_{ji}(t) w_h \left( f(p_i+p_j+l_h) - f(p_j+l_h) \right) \notag
\\\leq&
w_h \left(f(p_i+p_j+l_h)-f(p_i+l_h) \right).\label{eq:thm1:B}
\end{align}
To analyze the contribution of jobs $i,j$ we observe that by $i\prec_\ell j$ we have
\( 
   w_j < \min \phi_{ji} w_i
\) 
which implies
\begin{align}
&\min_t \phi_{ji}(t) w_i \left( f(a+p_i+p_j)-f(a+p_i+p_j+l_k) \right) \notag
\\<&
w_j \left( f(a+p_i+p_j)-f(a+p_i+p_j+l_k) \right) .\label{eq:thm1:ij}
\end{align}

Summing up \eqref{eq:thm1:B} for every $1 \leq h \leq k$ and \eqref{eq:thm1:ij} yields \eqref{eq:thm1:goal} as required, and completes the proof.
\end{myproof}

\subsection{Global order property for $\beta>1$}

\begin{theorem}\label{thm:lg-nikhil}
The implication $i\prec_\ell j \Rightarrow i\prec_g j$ holds when $\beta \geq 1$.
\end{theorem}

\begin{myproof}
By Theorem~\ref{thm:lg-big} it suffices to consider the case $p_j<p_i$.
Assume $i\prec_\ell j$ and consider a schedule $S$ of the form $AjBi$ for some job sequences $A,B$.

The proof is by induction on the number of jobs in $B$. The base case $B=\emptyset$ follows from $i\prec_\ell j$.
For the induction step, we assume that $A'jB'i$ is suboptimal for all job sequences $A',B'$ where $B'$ has strictly less jobs than $B$.
Formally we denote $B$ as the job sequence $1,2,\ldots,k$ for some $k\geq 1$.
If for some $1\leq h \leq k$ we have
\begin{equation*}
     F(AjBi) \geq F(A(12\ldots h)j(h+1,\ldots,k)i),
\end{equation*}
then by induction we immediately obtain that $AjBi$ is suboptimal.
Therefore we assume
\begin{equation}
\label{eq:eq1}
\forall 1\leq h\leq k :   F(AjBi) < F(A(12\ldots h)j(h+1,\ldots,k)i).
\end{equation}

Then we show that $F(AjBi) > F(AiBj)$ to establish sub-optimality of $AjBi$.

For the remainder of the proof, we introduce the following notations.
We denote by $a$ the total processing time of $A$.  In addition we use $h$ and $h'$ to index the jobs in $B$, and denote by $l_h$ the total processing time of jobs $1,2,\ldots,h$, and by $b=a+l_k$ the total processing time of $AB$.
We also introduce the expressions
$$\delta^i_{h}:= f(a+p_i + l_h) - f(a+l_h)$$
and
$$ \gamma^i_{h} :=  f(a+p_i+l_h) - f(a+p_i+l_{h-1})$$
and define $\delta^j_{h}, \gamma^j_h$ analogously.

Equations \eqref{eq:eq1} imply that
\begin{align}
 & \sum_{h'=1}^h w_{h'} \big(f(a+p_j + l_{h'}) - f(a+l_{h'})\big) \notag
 \\ \leq &
 w_j  \big(f(a+p_j+l_h) - f(a+p_j)\big) \notag \\
      = &
      w_j \sum_{h'=1}^h \big(f(a+p_j+l_{h'}) - f(a+p_j+l_{h'-1})\big)
     \label{cond}
    \end{align}
where we use the convention that $l_0=0$.

We restate \eqref{cond} as follows:  For each $1 \leq h \leq k$,
\begin{equation}
  \label{cond2}
  w_1 \delta^j_1 + \ldots + w_h \delta^j_h \leq w_j (\gamma^j_{1} + \ldots + \gamma^j_h)
\end{equation}

For  $a < b$
define $$g(x) = x \frac{f(b+x) -f(x+a)}{f(b+x) - f(b)}.$$
By Lemma \ref{inc} $g$ is non-decreasing in $x$.

We need to show that $F(AiBj) < F(AjBi)$. As $p_i > p_j$ by case assumption, when we move from $AjBi$ to $AiBj$, the completion times of $j$ and the jobs in $B$ increase and that of $i$ decreases. Thus the statement is equivalent to showing that
\begin{align}	\label{goal}
&  \sum_{h=1}^k w_h \biggl(f(a + p_i +l_h) - f(a + p_j +l_h)\biggl) \notag \\
< & \phantom{-} w_i \big(f(a+p_j+p_i+l_k ) - f(a+p_i )\big) \notag
\\
& -  w_j \big(f(a+p_j +p_i +l_k) - f(a+p_j) \big)
\end{align}

Now, by assumption $i\prec_\ell j$, it holds that
\begin{align*}
& w_j \big(f(a+p_j+p_i+l_k) - f(a+p_j+l_k)\big)
\\
< & w_i \big(f(a+p_i+p_j+l_k) - f(a+p_i+l_k)\big),
\end{align*}
 thus to show \eqref{goal}  it suffices to show that
\begin{align}
&\sum_{h=1}^k w_h \big(f(a + p_i +l_h) - f(a + p_j +l_h)\big) \notag
\\\leq &
\phantom{-}w_i \big(f(a + p_i+ l_k) - f(a+p_i )\big)
\notag\\
&- w_j \big(f(a+p_j +l_k) - f(a+p_j) \big).
\label{goal2}
\end{align}

We reformulate \eqref{goal2} as
 $$\sum_{h=1}^k w_h (\delta^i_h - \delta^j_h) \leq w_i  \sum_{h=1}^k \gamma^i_h   - w_j \sum_h \gamma^j_h .$$

Since $w_i \geq w_j p_i/p_j$ by Lemma~\ref{lem:local}, it suffices to show that for every $1 \leq h \leq k$,
\begin{equation}
\label{finalcond}
\sum_{h=1}^k w_h (\delta^i_h - \delta^j_h)
\leq w_j   \sum_{h=1}^k \biggl( \frac{p_i}{p_j} \gamma^i_h - \gamma^j_h \biggr) .
\end{equation}

We define for every job $1\leq h\leq k$
$$q_h:= \frac{\delta^i_{h}}{\delta^j_h} - \frac{\delta^i_{h+1}}{\delta^j_{h+1}}$$
where we use the convention $\delta^i_{k+1}/\delta^j_{k+1}=1$.
Note that by Corollary~\ref{lem:q}, all $q_h$ are non-negative.

Multiplying for a given $h$, equation~\eqref{cond2} by $q_h$, and summing over all $1\leq h \leq k$
we obtain
$$   \sum_{h=1}^k   w_h  \delta^j_h (\sum_{h'\geq h} q_{h'}) \leq w_j \sum_{h=1}^k  \gamma^j_h (\sum_{h'\geq h} q_{h'}). $$

As the sum over $q_{h'}$ telescopes, we can rewrite the above as
$$ \sum_{h=1}^k w_h (\delta^i_h - \delta^j_h) \leq  w_j \sum_{h=1}^k \gamma^j_h (\delta^i_h/\delta^j_h -1)  $$

Thus to prove \eqref{finalcond}, it suffices to show that
$$  \gamma^j_h (\delta^i_h/\delta^j_h -1) \leq  (p_i/p_j) \gamma^i_h -\gamma^j_h $$
or equivalently,
$$ \delta^i_h / \delta^j_h \leq (p_i /p_j) (\gamma^i_h /\gamma^j_h)$$

But this is exactly what Lemma~\ref{inc} gives us. In particular,
\[
\gamma^j_h/\delta^j_h p_j \leq p_i \gamma^i_h/\delta^i_h
\]
is equivalent to
\begin{align*}
&
p_j \frac{ f(a+p_j+l_h) - f(a+p_j+l_{h-1}) }{  f(a+p_j + l_h) - f(a+l_h)}
\\
\leq&
p_i \frac{ f(a+p_i+l_h) - f(a+p_i+l_{h-1}) }{  f(a+p_i + l_h) - f(a+l_h)}
\end{align*}
which follows Lemma~\ref{inc} since $p_{i} \geq p_{j}$.
Therefore, the theorem holds.
\end{myproof}

\subsection{Global order property for $0<\beta<1$ and $p_j\leq p_i$}

\begin{theorem}~\label{thm:lg-oscar}
The implication $i\prec_\ell j \Rightarrow i\prec_g j$ holds when $p_j\leq p_i$, $w_i/w_j \geq (p_i/p_j)^{2-\beta}$ and $0<\beta<1$.
\end{theorem}
\begin{myproof}
The proof is along the same lines as the previous one.  Hence we assume \eqref{eq:eq1} and need to show $F(AjBi)>F(AiBj)$,
and for this purpose show
\[
        F (A B j i) - F (A B i j) < F (A j B i) - F (A i B j),
\]
where the left hand side is positive by $i\prec_\ell j$.
Equivalently we have to show
\begin{equation}
 F (A B ji) - F (A j Bi) < F (A B ij) - F (A i Bj) . \label{eq:local:global}
\end{equation}

First we claim that
\begin{equation} \label{eq:wp1b}
  \frac{w_i}{w_j} > \left( \frac{p_i}{p_j} \right)^{1-\beta}.
\end{equation}
Indeed, from Lemma~\ref{lem:local} we have $w_j/w_i\leq \phi_{ji}(0)$.
By Lemma~\ref{lem:monotone}, the function $\phi_{ji}(t)$ is increasing, and
by Lemma ~\ref{lem:limit}, it is upper bounded by $p_j/p_i$. Hence $w_j/w_i \leq p_j/p_i$, and for $0<\beta<1$ and $p_j\leq p_i$ inequality~\eqref{eq:wp1b} follows.

Therefore by Lemma~\ref{lem:wf} we have
\begin{align}
\frac{w_i}{w_j} \frac{f(b+p_i) -f(a+p_i)}{f(b+p_j) -f(a+p_j)}
&>\notag\\
 \left( \frac{p_i}{p_j} \right)^{1-\beta} \frac{f(b+p_i) -f(a+p_i)}{f(b+p_j) -f(a+p_j)} &\geq 1.
 \label{eq:wiwj_pipj_1b}
 \end{align}

For convenience we denote
\[
\Delta(t) := \sum_{h \in B} w_h (f (a+l_h + t) - f(a+l_h)).
\]
We have
  \begin{align*}
    0<&F (A B j i) - F (A j B i) \\
= & w_j (f (b + p_j) - f (a + p_j)) - \Delta(p_j)\nonumber\\
     <& \phantom{\cdot} \frac{w_i}{w_j}\frac{f (b + p_i) - f (a + p_i)}{f (b + p_j) - f (a + p_j)} \cdot
   \notag\\& \cdot \big(w_j (f (b + p_j) - f (a + p_j)) - \Delta(p_j) \big)                 \nonumber\\
     =& \phantom{-} w_i \big(f (b + p_i) - f (a + p_i)\big)
     \notag\\& - \frac{w_i}{w_j} \frac{f
    (b + p_i) - f (a + p_i)}{f (b + p_j) - f (a + p_j)} \Delta(p_j) \nonumber\\
     < & w_i \big(f (b + p_i) - f (a + p_i)\big)
     \notag\\&-
      \frac{w_i}{w_j} \frac{f (b + p_i) - f (a + p_i)}{f (b +    p_j) - f (a + p_j)} \min_{t \geq 0} \frac{f (t + p_j) - f (t)}{f (t +    p_i) - f (t)} \Delta(p_i).
  \end{align*}
  The first inequality follows from assumption~\eqref{eq:eq1} with $h=k$.
  The second inequality follows from \eqref{eq:wiwj_pipj_1b}.
  The third inequality holds since $f(t+p_j) < f(t+p_i)$ for all $t\geq 0$.

In order to upper bound the latter expression by
  \begin{align*}
      &w_i (f (b + p_i) - f (a + p_i)) - \Delta(p_i)
      \\= & F (A B i j) - F (A i B j)
  \end{align*}
  as required, it suffices to show
  \begin{align*}
    \frac{w_i}{w_j}\frac{f (b + p_i) - f (a + p_i)}{f (b + p_j) - f
    (a + p_j)} \min_{t \geq 0} \frac{f (t + p_j) - f (t)}{f (t + p_i) - f
    (t)} \geq& 1.
  \end{align*}
By $0<\beta<1$ and Corollary~\ref{lem:q} the fraction
  $\frac{f (t + p_j) - f (t)}{f (t + p_i) - f
    (t)}$ is decreasing, and its limit  when $t \rightarrow \infty$ is $p_j / p_i$, by the same analysis as in the proof of Lemma~\ref{lem:limit}. Therefore $\min_{t \geq 0} \frac{f (t + p_j) - f (t)}{f (t + p_i) - f
    (t)}  \geq \frac{p_j}{p_i}$. Hence
  \begin{align*}
    &\frac{w_i}{w_j} \frac{f (b + p_i) - f (a + p_i)}{f (b + p_j)
    - f (a + p_j)} \min_{t \geq 0} \frac{f (t + p_j) - f (t)}{f (t + p_i)
    - f (t)}
    \\
    \geq &
    \frac{w_i}{w_j} \frac{f (b + p_i) - f (a + p_i)}{f (b + p_j)
    - f (a + p_j)}
    \frac{p_j}{p_i}
    \\
    \geq &
\left(    \frac{p_i}{p_j} \right)^{2-\beta} \frac{f (b + p_i) - f (a + p_i)}{f (b + p_j)
    - f (a + p_j)}
    \frac{p_j}{p_i}
    \\
    =&
\left(    \frac{p_i}{p_j} \right)^{1-\beta} \frac{f (b + p_i) - f (a + p_i)}{f (b + p_j)
    - f (a + p_j)}
    \geq 1.
  \end{align*}
  where the second inequality follows from the theorem hypothesis and the last inequality from  Lemma~\ref{lem:wf}. This concludes the proof
  of the theorem. 
\end{myproof}

\section{Generalization}
\label{sec:generalization}

We can provide some technical generalizations of the aforementioned theorems.  For any pair of jobs $i,j$, and job sequence $T$ of total length $t$, we denote by $i\prec_{\ell(t)}j$ the property $F(Tij) < F(Tji)$.  Now suppose that none of  $i\prec_\ell j$ or $j\prec_\ell j$ holds, and say $p_i > p_j$ and $\beta>1$. Then from Lemma~\ref{lem:monotone} it follows that there exist a unique time $t$, such that for all $t' < t$  we have $i\prec_{\ell(t')}j$ and for all $t' > t$ we have $j\prec_{\ell(t')}i$.  These properties are denoted respectively by $i\prec_{\ell[0,t)} j$ and  $j\prec_{\ell(t,\infty)} j$.  In case $p_i < p_j, \beta>1$ or $p_i > p_j, 0<\beta<1$, we have the symmetric situation $j\prec_{\ell[0,t)} i$ and $i\prec_{\ell(t,\infty)} j$.

This notation can be extended also to the global order property.  If for every job sequences $A,B$ with $A$ having total length at least $t$ we have $F(AiBj) < F(AjBi)$, then we say that $i,j$ satisfy the global order property in the interval $(t,\infty)$ and denote it by $i\prec_{g(t,\infty)}j$.  The property $i\prec_{g[0,t)}j$ is defined similarly for job sequences $A,B$ of total length at most $t$.

The proof of Theorem~\ref{thm:lg-nikhil} actually shows the stronger statement: if $\beta>1$ and $p_i\geq p_j$, then $j\prec_{\ell(t,\infty)} i$ implies $j\prec_{g(t,\infty)} i$.  The same implication does not hold for interval $[0,t)$, as shown by the following counter example.
It
consists      of    a   3-job   instance     for      $\beta=2$     with
$p_i=13,w_i=7,p_j=8,w_j=5,p_k=1,w_k=1$.    For   $t=19/18$,  we   have
$i\prec_{\ell[0,t)}j$ and  $j\prec_{\ell(t,\infty)}i$.  But the unique
optimal solution  is the  sequence $jki$, meaning  that we  do not have
$i\prec_{g[0,t)}j$.

These generalizations can be summarized as follows.
\\
\begin{center}
\begin{tabular}{|c|c|c|}
\hline
& $p_i\leq p_j$ & $p_i\geq p_j$
\\
& $0<\beta<1$ & $\beta>1$
\\ \hline
$i \prec_{\ell[0,t)} j \Rightarrow i \prec_{g[0,t)} j$ & Yes & No
\\ \hline
$j \prec_{\ell(t,\infty)} i \Rightarrow j \prec_{g(t,\infty)} i$ & Open & Yes
\\ \hline
\end{tabular}
\end{center}

\section{Experimental study}
\label{sec:experimental}
We  conclude this  paper with  an experimental  study,  evaluating the
impact of the proposed rules on the performance of a search procedure.
The experiments are based on a C++ program executed on a GNU/Linux machine with 3 Intel Xeon processors, each with 4 cores, running at 2.6Ghz and 32Gb RAM. In order to be independent on the machine environment, we measured the number of generated search nodes rather than running time.
Hence we use a timeout which is not expressed in seconds, but in time units corresponding to the processing of a search node by the program.
Note that we use the rules that we have proved (not the ones in the conjecture).
Following              the              approach             described
in~\citep{HohnJacobs:12:Experimental-analytical-quadratic-penalty-function},   we  consider
the Algorithm  A* by \citet{HartNilssonBertram:72:Proposition-A-star-algorithm}.

The search  space is the
directed     acyclic    graph     consisting     of    all     subsets
$S\subseteq\{1,\ldots,n\}$. Note  that the potential  search space has
size $2^n$ which is already less  than the space of the $n!$ different
schedules.  In  this graph  for every  vertex $S$ there  is an  arc to
$S\backslash\{j\}$ for any  $j\in S$. It is labeled  with $j$, and has
cost  $w_j t^\beta$ for  $t=\sum_{i\in S}  p_i$.  Every  directed path
from the root $\{1,\ldots,n\}$ to the target $\{\}$ corresponds to a schedule of
an objective value being the total arc cost.

The algorithm A* finds a shortest path from the root to the target vertex, and as Dijkstra's algorithm uses a priority queue to select the next vertex to explore.  But the difference of A* is that it uses as weight for vertex $u$ not only the distance from the source to $u$, but also a lower bound on the distance from $u$ to the destination.
A set $S$ is maintained containing all vertices $u$ for which a shortest path has already been discovered.  Initially $S=\{s\}$ for the root vertex $s$.  In Dijkstra's algorithm the priority queue contains all remaining vertices $v$, with the priority $\min_{u\in S} d(s,u)+w(u,v)$, where $w(u,v)$ is the weight of the arc $(u,v)$.  However in the algorithm A* this priority is replaced by $\min_{u\in S} d(s,u)+w(u,v) + h(v)$, where $h$ is some lower bound on the distance from $v$ to the target.  This function should satisfy $h(v)=0$ if $v$ is the target and $h(u) \leq w(u,v)+h(v)$ for every arc $(u,v)$. The function $h$ used in our experiments satisfies these properties.


Pruning is  done when constructing the  list of outgoing  arcs at some
vertex $S$.  Potentially  every job $j\in S$ can  generate an arc, but
order properties might prevent that.   Let $i$ be the label of the
arc   leading  to   $S$  (assuming   $S$  is   not  the   root).   Let
$t_1=\sum_{k\in S} p_k$.  We distinguish two kind of pruning rules:
\begin{description}
    \item[\textbf{Arc pruning.}] The arc from $S$ to $S\setminus\{j\}$ for $j\in S$ is pruned if $i\prec_{\ell(t_1-p_j)}j$, because placing job $j$ adjacent to $i$ at this position would be suboptimal.
\item[\textbf{Vertex pruning.}] All arcs leaving vertex $S$ are pruned, if there is a
job $j\in  S$ with $i \prec_{g[0,t_1]} j$, as again placing job $j$ somewhere before job $i$  would be suboptimal.
\end{description}
In the lack of a complete characterization of the global precedence relation, we have to weaken the vertex pruning rule by replacing $i \prec_{g[0,t_1]} j$ with a condition implying $i\prec_g j$.   These would be the Sen-Dileepan-Ruparel condition in general or for $\beta=2$ the Mondal-Sen-Höhn-Jacobs conditions.  Our pruning rules consist of using our conditions for global precedence.

In a search tree  such an arc  pruning would cut  the whole subtree  attached to
that arc,  but in a directed acyclic graph (DAG) the improvement  is not so
significant. As the typical in-degree of a vertex is linear in $n$, a linear number of
arc-cut is necessary to remove a vertex from the DAG.

Figure~\ref{fig:dag} illustrates the DAG  explored by A* for $\beta=2$
on the instance consisting  of the following (processing time, priority
weight) pairs  :
\[
(10,5),(10,5), (11,3),(13,6),(8,4),(12,6).
\]
Arcs  are labeled  with their cost.  The last two arcs have the
same weight, as the lower bound on single job sets is tight.
\begin{figure}[htb]
  \centering
  \includegraphics[width=0.5\textwidth]{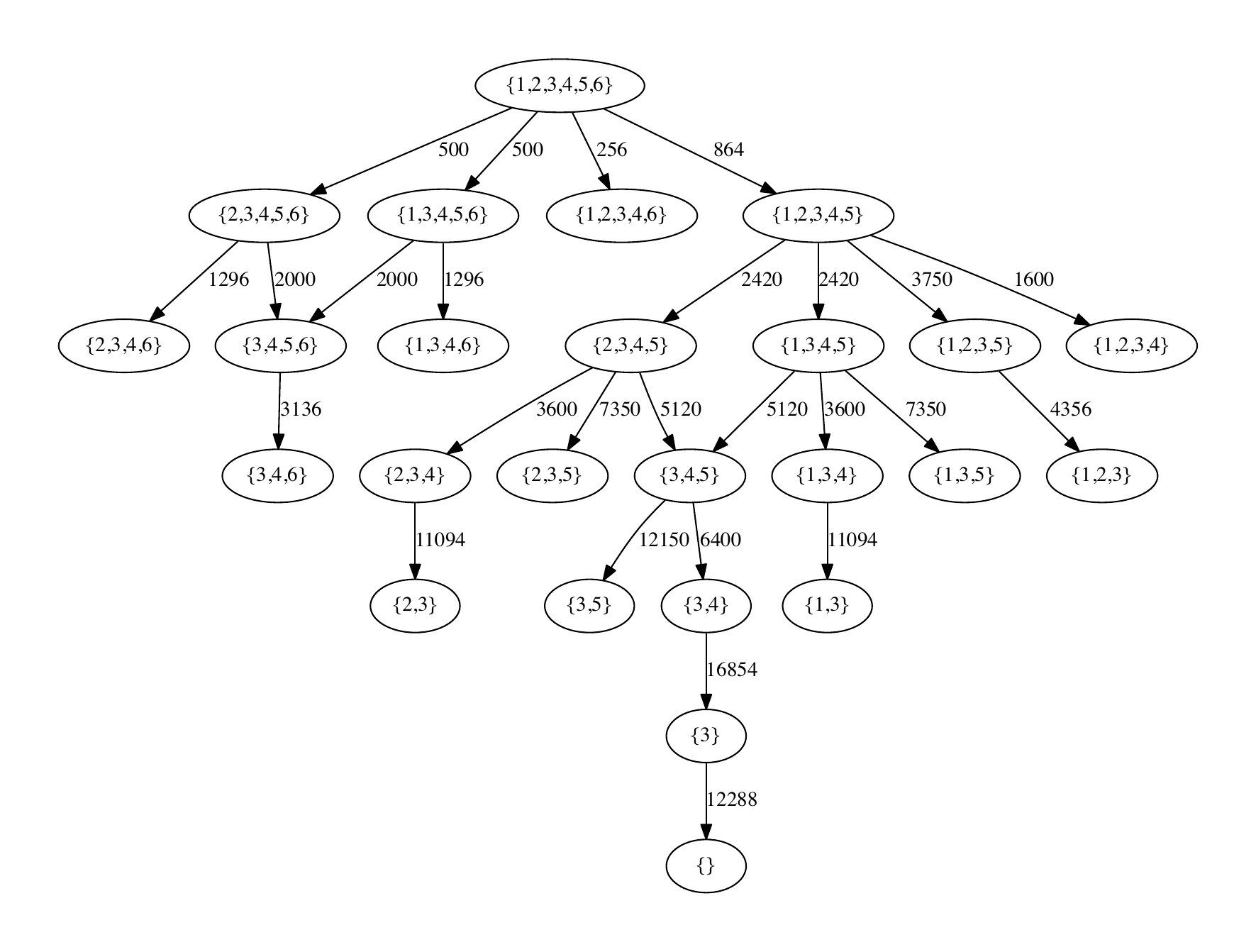}
  \caption{Example of the portion of the search graph explored by the algorithm A*.}
  \label{fig:dag}
\end{figure}

A  simple  additional pruning  could be  done  when  remaining jobs  to  be scheduled form a \emph{trivial} subinstance.  By this we mean that all pairs  of jobs  $i,j$ from  this subinstance  are comparable  with the order  $\prec_{\ell[0,t_1]}$.  In  that  case the local order  is actually  a total order,  which describes in  a simple manner the  optimal schedule  for this subinstance.   In that  case we could simply generate a  single path from the node $S$  to the target vertex $\{\}$.  However experiments showed that detecting this situation is too costly compared with the benefit we could gain from this pruning rule.

\subsection{Random instances}
\label{subsec:gener}

We adopt the model of random instances described by Höhn and Jacobs.
Most previous  experimental results were made  by generating processing
times and  weights uniformly from  some interval, which leads  to easy
instances, since  any job pair $i,j$ satisfies  with probability $1/2$
the Sen-Dileepan-Ruparel condition,  i.e.\  $i\prec_g  j$  or  $j\prec_g  i$.   As  an
alternative, \citet{HohnJacobs:12:Experimental-analytical-quadratic-penalty-function}  proposed a random model, where the Smith-ratio  of a job is selected according to
$2^{N(0,\sigma^2)}$ with $N$ being the normal distribution centered at
$0$ with  variance $\sigma$.  Therefore for  $\beta=2$ the probability
that  two jobs  satisfy the Mondal-Sen-Höhn-Jacobs-2 condition depends on $\sigma$,
as it compares the Smith-ratio among the jobs.


We adopted their  model for other values of  $\beta$ as follows.  When
$\beta>1$,  the  condition  for  $i\prec_g  j$ of  our  conditions
can  be approximated, when $p_j/p_i$ tends to infinity, by the relation $w_i/p_i
\geq  \beta  w_j/p_j$.   Therefore   in  order  to  obtain  a  similar
``hardness'' of  the random instances for the  same parameter $\sigma$
for different values of $\beta>1$, we choose the Smith-ratio according
to  $2^{N(0,\beta^2\sigma^2)}$.   This   way  the  ratio  between  the
Smith-ratios of  two jobs is  a random variable from  the distribution
$2^{2N(0,\beta^2\sigma^2)}$, and the probability that this value is at
least $\beta$ depends only on $\sigma$.

However when  $\beta$ is  between $0$ and  $1$, the  our condition for
$i\prec_g j$ of  our rule can be approximated  when $p_j/p_i$ tends to
infinity by  the relation $w_i/p_i  \geq 2 w_j/p_j$, and  therefore we
choose   the    Smith-ratio   of    the   jobs   according    to   the
$\beta$-independent distribution $2^{N(0,4\sigma^2 )}$.

The instances  of our  main test sets  are generated as  follows.  For
each      choice      of      $\sigma\in\{0.1,0.2,\ldots,1\}$      and
$\beta\in\{0.5,0.8,1.1,\ldots,3.2\}$, we  generated $25$ instances of
$20$  jobs  each.  The  processing  time  of  every job  is  uniformly
generated  in  $\{1,2,\ldots,100\}$.   Then  the weight  is  generated
according to the above  described distribution.  Note that the problem
is independent from scaling of  processing time or  weights, motivating
the arbitrary choice of the constant $100$.

\subsection{Hardness of instances}
\label{subsec:hardness}

As a measure of the hardness  of instances, we consider the portion of
job  pairs $i,j$  which  satisfy global
precedence.  By this we mean that we have either $i\prec_{g[0,t_1]} j$
or $j  \prec_{g[0,t_1]} i$  for $t_1$ being  the total  processing time
over all jobs excepting jobs $i,j$. Figure~\ref{fig:global} shows this
measure for various choices of $\beta$.

The results depicted in  Figure~\ref{fig:global} confirm the choice of
the model  of random instances.  Indeed the hardness of  the instances
seems to  depend only  little on $\beta$,  except for  $\beta=2$ where
particularly strong precedence rules have been established.  In addition
the impact  of our new  rules is significant, and  further experiments
show how  this improvement influences  the number of  generated nodes,
and therefore the running time.  Moreover it is quite visible from the
measures that the instances are  more difficult to solve when they are
generated with a small $\sigma$ value.

\subsection{Comparison between forward and backward approaches}
\label{subsec:forward-backward}

In  this section,  we consider  a  variant of the
algorithm. The algorithm described so far is called the \emph{backward} approach, and the variant is called the \emph{forward}  approach. Here a  partial  schedule
describes  a  prefix of  length  $t$ of  a  complete  schedule and  is
extended to  its right along an edge  of the search tree,  and in this
variant the  basic lower bound  is $h(S) := \sum_{i\in S}  w_i (t+p_i)^\beta$.
However  in  the  \emph{backward}  approach, a  partial  schedule  $S$
describes a suffix of a complete schedule and is extended to its left.
For this variant, we choose $h(S) := \sum_{i\in S}  w_i p_i^\beta$.
\citet{KaindlKainzRadda:01:Asymmetry-in-search} give
experimental  evidence that  the backward  variant generates  for some
problems less  nodes in the search  tree, and this fact  has also been
observed           by  \citet{HohnJacobs:12:Experimental-analytical-quadratic-penalty-function}.

We conducted an experimental study  in order to find out which variant
is  most  likely to  be  more efficient.   The  results  are shown  in
Figure~\ref{fig:direction}.   The values  are  most significative  for
small $\sigma$ values,  since for large values the  instances are easy
anyway  and the  choice of  the variant  is not  very  important.  The
results indicate that \emph{without} our  rules the forward variant should be
used only when $\beta<1$ or $\beta=2$, while \emph{with} our rules the forward
variant should be used only when $\beta>1$.

Later on, when we measured the impact of our rules in the subsequent experiments,
we compared  the behavior  of the algorithm  using the  most favorable
variant dependent on the value of $\beta$ as described above.

\subsection{Timeout}
\label{subsec:timeout}

During  the resolution  a timeout  was set,  aborting  executions that
needed more than a million nodes.  In Figure~\ref{fig:solved} we show
the fraction  of instances  that could be  solved within  the limited
number of nodes.  From these experiments we measure the instance sizes
that  can be efficiently  solved, and  observe that  this limit  is of
course  smaller  when  $\sigma$  is  small, as  the  instances  become
harder. But  we also  observe that  with the usage  of our  rules much
larger instances can be solved.

When $\beta$ is close to $1$,  and instances consist of jobs of almost
equal Smith-ratio,  the different  schedules diverge only  slightly in
cost, and  intuitively one has to  develop a schedule  prefix close to
the  makespan, in order  to find  out that  it cannot lead  to the
optimum.   However for  $\beta=2$,
the Mondal-Sen-H\"ohn-Jacobs  conditions
make  the  instances easier  to solve  than  for  other values  of
$\beta$, even  close to $2$.  Note  that we had  to consider different
instance sizes,  in order  to obtain comparable  results, as  with our
rules all 20 job instances could be solved.

\begin{figure}[ht]
  \centering
  \includegraphics[width=8cm]{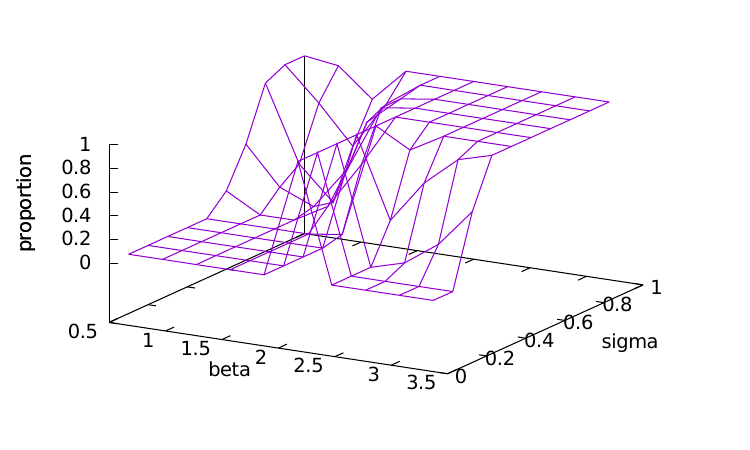}
  \includegraphics[width=8cm]{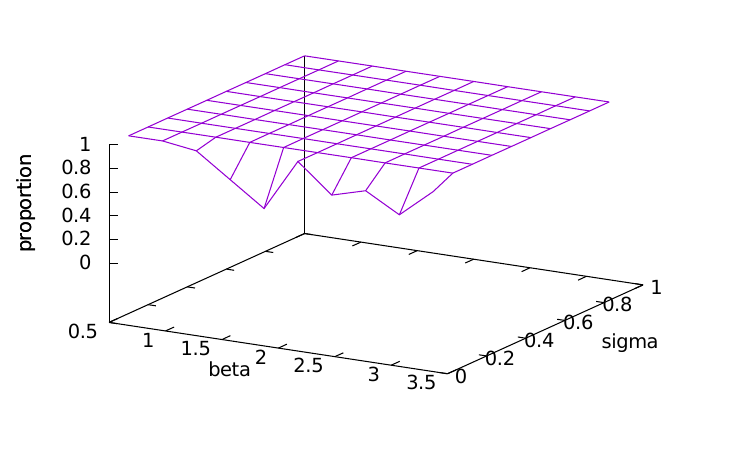}
  \caption{Proportion of  instances which  could be solved  within the
    imposed time  limit of a  million nodes, with (below)  and without
    (above) the new rules. For every $(\beta,\sigma)$, the evaluation is done on 25 instances each consists of 40 random jobs.}
  \label{fig:solved}
\end{figure}

\subsection{Improvement factor}
\label{subsec:improv}

In  this section  we  measure the  influence  on the  number of  nodes
generated during a resolution when  our rules are used.  For $\beta=2$
we  compare our  performance with  the  Mondal-Sen-H\"ohn-Jacobs conditions
defined in \citep{MondalSen:00:Conjecture-order-constraint-quadratic-penalty-function}
and proved in \citep{HohnJacobs:12:Experimental-analytical-quadratic-penalty-function},
while   for   other   values   of   $\beta$  we   compare   with   the
Sen-Dileepan-Ruparel condition defined in \citep{SenDileepan:90:Order-constraint-generalized-quadratic-penalty-function}.
For fairness we  excluded instances where
the   timeout   was   reached   without   the  use   of   our   rules.
Figure~\ref{fig:ratio} shows  the ratio between the  average number of
generated nodes when the algorithm is  run with our rules, and when it
is  run  without  our  rules.   Clearly this  factor  is  smaller  for
$\beta=2$, since the Mondal-Sen-H\"ohn-Jacobs conditions apply here.

We  observe that  the improvement  factor is  more significant  for hard
instances, i.e.\  when $\sigma$ is  small.  From the figures  it seems
that this behavior is not monotone, for $\beta=1.1$ the factor is less
important with  $\sigma=0.1$ than with $\sigma=0.3$.   However this is
an  artifact of our  pessimistic measurements,  since we  average only
over instances which could be solved within the time limit, so in the
statistics we filtered out the really hard instances.

\section{Performance measurements for $\beta=2$}
\label{sec:performance_beta_2}

For $\beta=2$, \citet{HohnJacobs:12:Experimental-analytical-quadratic-penalty-function}
provide  several test sets  to measure  the impact  of their  rules in
different variants,  see \citet{Hohn:12:Library-quadratic-penalty-function}.  For  completeness we
selected two data sets from their collection to compare our rules with
theirs.

The first set called \texttt{set-n}  contains for every number of jobs
$n=1,2,\ldots,35$, 10 instances generated with parameter $\sigma=0.5$.
This test set permits to measure the  impact of our rules as a function of
the instance size.

The second  test set that  we considered is called  \texttt{set-T} and
contains  3  instances   of   25  jobs for  every parameter
\[
\sigma=0.100,  0.101, 0.102,  \ldots, 1.000.
\]

Results   are  depicted in Figure~\ref{fig:library}, and show an improvement in the range of one order of magnitude.

\section{Performance depending on input size}
\label{sec:performance_inputsize}

In addition we  show the performance of the  algorithm with our rules,
in dependence on the  number of jobs. Figure~\ref{fig:nodesByN2} shows
for different  number of jobs  the number of generated  nodes averaged
over 100 instances generated with different $\sigma$ parameters,
exposing  an  expected running  time  which  strongly  depends on  the
hardness of the instances.

\section{Conclusion}

We formulated the local global conjecture for the single machine scheduling problem of minimizing $w_j C_j^\beta$ for any positive constant $\beta$. We proved it for $\beta\geq 1$ substantially extending and improving over previous partial results.
We also show some partial results for the remaining case $0<\beta<1$.

We conducted experiments and measured the impact of our conditions on the running time (number of generated nodes) by an A* based exact resolution.  Improvements by a factor up to 1e4 have been observed.

Based on extensive experiments we believe that the conjecture should also hold in this case. However,
it seems to be substantially more complicated and new analytical techniques seem to be necessary.
We also describe a more general class of functions for which our results hold.
Determining the class of objective functions for which the local global conjecture holds would also be a very interesting direction to explore.

\paragraph{Acknowledgements}
  We are grateful to the anonymous referees who spotted errors in previous versions of this paper. This paper was supported by the PHC Van Gogh grant 33669TC, the FONDECYT grant 11140566, the NWO grant 639.022.211 and the ERC consolidator grant 617951.

\bibliographystyle{spbasic}

\newcommand{\plot}[3]{$\beta=#2.#3$&\begin{minipage}{6cm}\includegraphics[width=\linewidth]{#1#2#3.pdf}\end{minipage}}

\begin{figure*}[p]
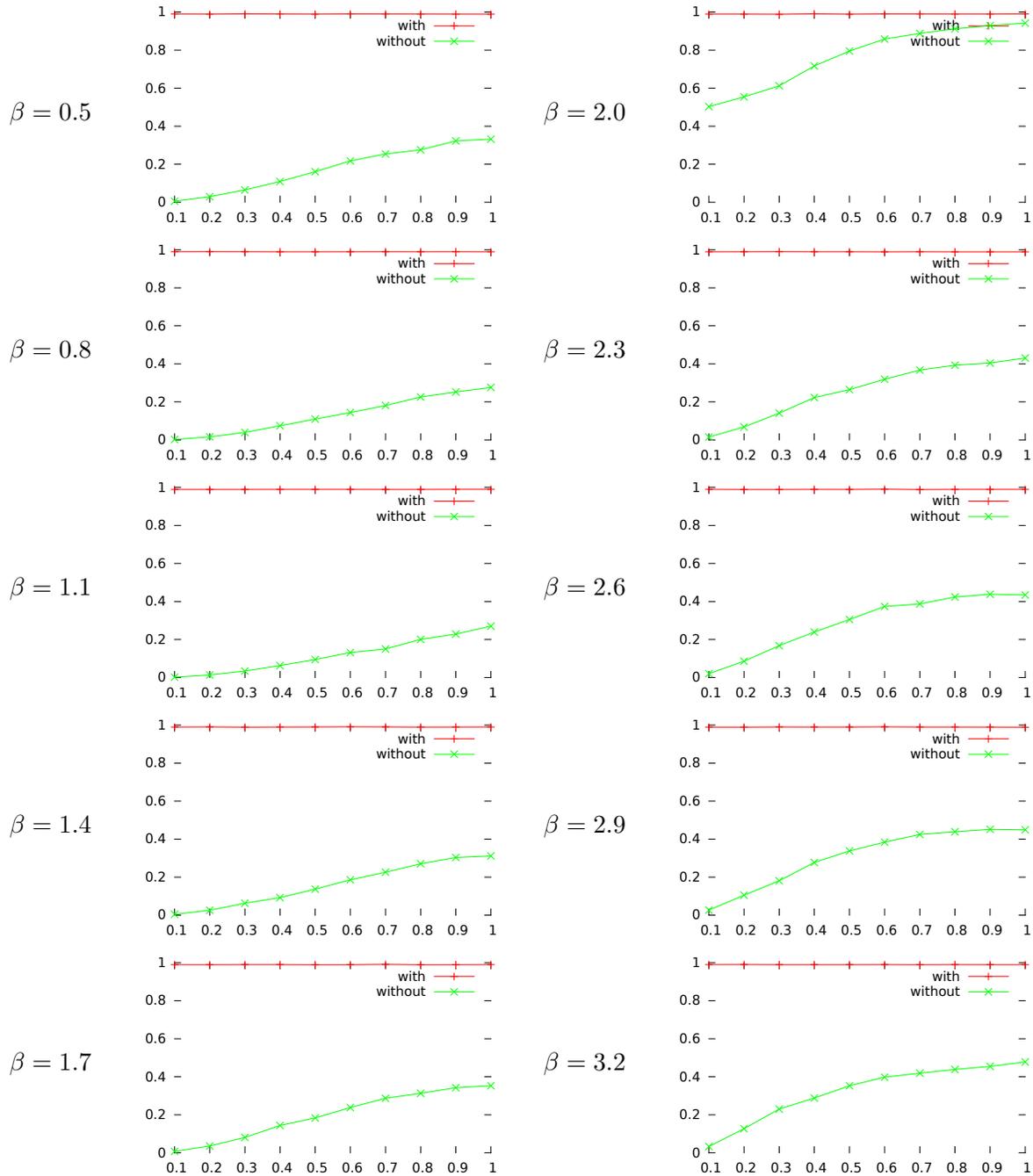

  \centering
  \begin{tabular}{llll}
    \plot{global60-}05 & \plot{global60-}20\\
    \plot{global60-}08 & \plot{global60-}23\\
    \plot{global60-}11 & \plot{global60-}26\\
    \plot{global60-}14 & \plot{global60-}29\\
    \plot{global60-}17 & \plot{global60-}32\\
  \end{tabular}
  \caption{Proportion of  job pairs (vertical axis) that satisfy a  global precedence
    relation as function of the  parameter $\sigma$ (horizontal axis). For every $(\beta,\sigma)$ pair, 25 instances, each contains 60 jobs, have been tested (Sections \ref{subsec:gener} and \ref{subsec:hardness}).}
  \label{fig:global}
\end{figure*}

\begin{figure*}[p]
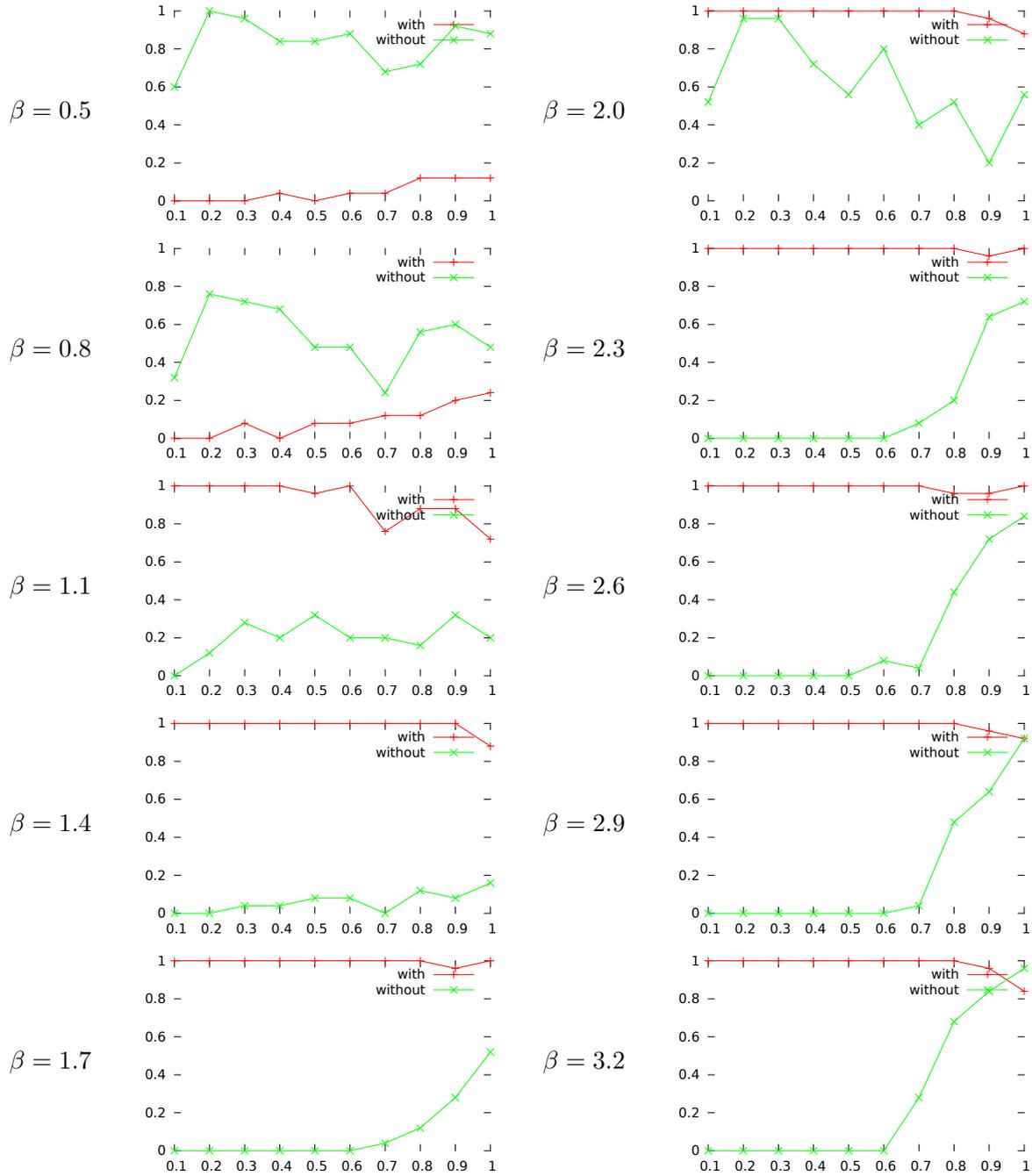

  \centering
  \begin{tabular}{llll}
    \plot{dir}05 & \plot{dir}20\\
    \plot{dir}08 & \plot{dir}23\\
    \plot{dir}11 & \plot{dir}26\\
    \plot{dir}14 & \plot{dir}29\\
    \plot{dir}17 & \plot{dir}32\\
  \end{tabular}
  \caption{Proportion  of  instances (vertical axis)  for  which the  forward  variant
    generated less  nodes than the  backward variant.  The  values are
    plotted as function of $\sigma$ (horizontal axis) with and without our
    new rules (Section \ref{subsec:forward-backward}). For every $(\beta,\sigma)$ pair,
    25 instances, each contains 20 jobs, have been tested.}
  \label{fig:direction}
\end{figure*}

\begin{figure*}[p]
  \centering
  \begin{tabular}{llll}
    \plot{ratio20-}05 & \plot{ratio20-}20\\
    \plot{ratio20-}08 & \plot{ratio20-}23\\
    \plot{ratio20-}11 & \plot{ratio20-}26\\
    \plot{ratio20-}14 & \plot{ratio20-}29\\
    \plot{ratio20-}17 & \plot{ratio20-}32\\
  \end{tabular}
  \caption{Average improvement factor (vertical axis) as
    function of $\beta$ and $\sigma$ (horizontal axis).
    For every $(\beta,\sigma)$ pair,
    25 instances, each contains 20 jobs, have been tested (Section \ref{subsec:improv}).}
  \label{fig:ratio}
\end{figure*}

\begin{figure*}[p]
  \centering
  \includegraphics[width=6cm]{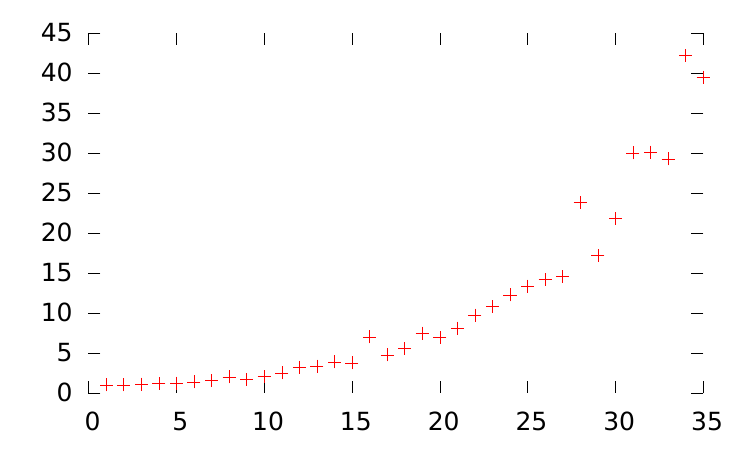}
  \includegraphics[width=6cm]{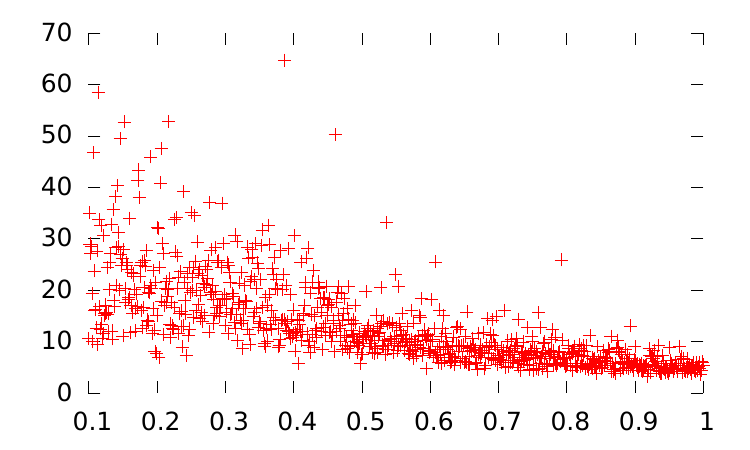}
  \caption{Improvement ratio for test sets \texttt{set-n} (left) and \texttt{set-T} (right) (Section \ref{sec:performance_beta_2}).}
  \label{fig:library}
\end{figure*}

\begin{figure*}[p]
  \centering
  \includegraphics[width=6cm]{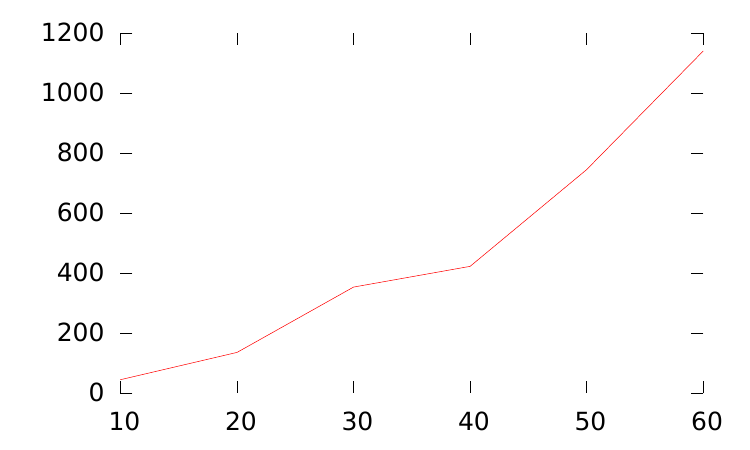}
  \includegraphics[width=6cm]{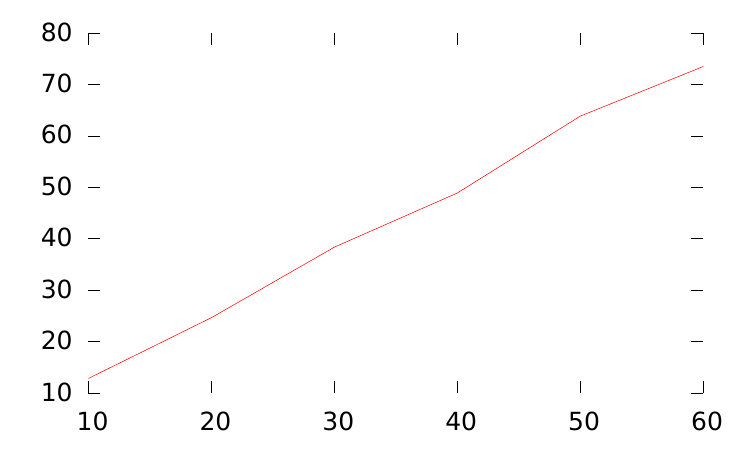}
  \caption{Average number of nodes (vertical axis) in  dependance on the size (horizontal axis)  of the
    instances,   generated   with $\beta = 2$ and  $\sigma=0.1$   on   the  left   and
    $\sigma=0.5$ on the right (Section \ref{sec:performance_inputsize}).}
  \label{fig:nodesByN2}
\end{figure*}

\end{document}